\documentclass{amsart}

\usepackage{graphicx}
\usepackage{epstopdf}
\usepackage{times}
\usepackage{amsmath}
\usepackage{amssymb}
\usepackage{amsthm}
\usepackage{mathabx}
\usepackage{float}
\usepackage{caption}
\usepackage{epsfig}
\usepackage{color}
\usepackage{algorithm}
\usepackage{subfig}
\usepackage{algorithmic}
\usepackage{enumerate}
\newcommand{\Nn}{{\mathbb{N}}}

\newcommand{\Gg}{{\mathcal{G}}}

\newcommand{\lb}{\left(}
\newcommand{\rb}{\right)}

\newcommand{\Mm}{\mathcal{M}}
\newcommand{\Bb}{\mathcal{B}} 
 \newtheorem{lemma}{Lemma}

\newtheorem{theorem}{Theorem}

\theoremstyle{definition}

\theoremstyle{remark}

\numberwithin{equation}{section}

\begin{document}

\title{Robust reputation-based ranking on multipartite rating networks}

\author{Jo\~{a}o Sa\'{u}de$^\ast$}
\address{Department of Electrical and Computer Engineering, Carnegie Mellon University, Pittsburgh, PA 15213}
\curraddr{LARSyS, Instituto Superior T\'ecnico,
University of Lisbon, Lisbon, Portugal}
\email{jsaude@andrew.cmu.edu}
\thanks{$^\ast$ The first two authors contributed equally to this work.\\
The work was partially supported through the Carnegie Mellon/Portugal Program managed by ICTI from FCT and by FCT grant SFRH/BD/52162/2013.}

\author{Guilherme Ramos$^\ast$}
\address{Department of Mathematics, Instituto Superior T\'ecnico,
University of Lisbon, Lisbon, Portugal}
\curraddr{Instituto de Telecomunica\c{c}\~oes, Instituto Superior T\'ecnico,
University of Lisbon, Lisbon, Portugal}
\email{guilherme.ramos@tecnico.ulisboa.pt}
\thanks{This work was developed under the scope of R\&D Unit 50008, financed by the applicable financial framework (FCT/MEC through national funds and when applicable co-funded by FEDER - PT2020 partnership agreement).
The second author acknowledges the support of the DP-PMI and Funda\c{c}\~ao para a Ci\^encia e a Tecnologia (Portugal), namely through scholarship SFRH/BD/52242/2013.}

\author{Carlos Caleiro}
\address{Department of Mathematics, Instituto Superior T\'ecnico,
University of Lisbon, Lisbon, Portugal}
\email{guilherme.ramos@tecnico.ulisboa.pt}
\thanks{}

\author{Soummya Kar}
\address{Department of Electrical and Computer Engineering, Carnegie Mellon University, Pittsburgh, PA 15213}
\email{soummyak@andrew.cmu.edu}
\thanks{}

%    General info
%\subjclass[2000]{Primary 54C40, 14E20; Secondary 46E25, 20C20}

\date{February 19, 2017 and, in revised form, April 27, 2017.}


\keywords{Trust and Reputation, Learning Preferences or Rankings, Data Mining and Classification}

\begin{abstract}
	The spread of online reviews, ratings and opinions and its growing influence on people's behavior and decisions boosted the interest to extract meaningful information from this data deluge.
Hence, crowdsourced ratings of products and services gained a critical role in business, governments, and others.
We propose a new reputation-based ranking system utilizing multipartite rating subnetworks, that clusters users by their similarities, using Kolmogorov complexity.
Our system is novel in that it reflects a diversity of opinions/preferences by assigning possibly distinct rankings, for the same item, for different groups of users.
We prove the convergence and efficiency of the system and show that it copes better with spamming/spurious users, and it is more robust to attacks than state-of-the-art approaches.
%	To allow diversity of opinions/preferences our system allows for different rankings, for the same item, for different groups of users.
%	This allows to explore a multitude of tastes and use this information to generate more meaningful rankings of items, which by turn are displayed to users accordingly to the cluster they belong.
%	Rating systems allow users to express their opinions about objects and services, but to make useful this information needs to be summarized. 
%	Ranking systems play a central role in our quotidian and spread across the Internet, as a mean of advertising and as a way to provide detailed informations from users to the service providers.
\end{abstract}

\maketitle

\section{Introduction} % (fold)
\label{sec:introduction} 
	Nowadays electronic commerce, streaming media and the collaborative economy (such as car rides and accommodation) are ubiquitous in our daily life.
	People's opinion can be as effective as an advertisement.
	% In the past, the spread of word was the best way to advertise a product/service, and now the same happens by different means.
	This promoted the development of crowdsourced ratings and reviews.
	Consumers started to use and rely on this information to decide whether or not to buy a product/service, have a meal on a certain restaurant, or simply attend an event, see~\cite{sparks2011impact,forman2008examining}.
	Aware of how ratings of products/services can impact sales, \cite{chevalier2006effect}, the sellers, in turn, rely on the ratings and reviews of their products to assess their commercial viability as well as to predict sales, \cite{dellarocas2007exploring}.
	Further, they use this information not only to improve their products, but also to target advertisement campaigns.
	Online ratings and reviews are perceived so important that it is desirable to detect and automatically correct rating manipulations through fake users' ratings.  
	
	\textbf{Previous work.} A simple way to collect and process the huge amount of ratings is using the \emph{arithmetic average} (AA).
	Some of the drawbacks of the AA are the indistinguishability of users, so it treats the most relevant raters and spam in the same way. 
	Therefore AA is prone to manipulation of ratings through malicious attacks or spamming.
	Further, AA might be misleading, because it does not capture the possible multimodal behavior of ratings, see~\cite{hu2006can}.
	For instance, in a bimodal distribution of ratings on the opposite extremes, the average would be located in the middle where the density of votes of users is low.

	Using weighted average algorithms allows us to attribute different importance to certain users. %, and filter some noise.
	This was explored in previous works \cite{yu2006decoding,de2010iterative}.
	The authors in \cite{li2012robust} used a modification of the weighted average.
	%** ranking systems use a reputation based systems, where more relevant users weight more on the final ranking of an item.
	In~\cite{mizzaro2003quality}, the author uses an additional time-dependent quantity to weigh the ratings of users.
	These methods are more robust to spamming and attacks than the AA.

	%It also incorporates the reputation of users based on aggregate differences between them.
	The methods above have a bipartite graph structure. 
	There are two types of nodes, users and items, and weighted edges (ratings) linking the two, see Figure~\ref{fig:grafo}, when not considering the dashed lines.	
	%In the bipartite rating systems we compute the reputations of users based on the divergence of their ratings and a weighted sum of rankings of the products.
	This kind of algorithm does not take into account possible relations between users or users' preferences.
	Furthermore, these approaches do not incorporate the (possible) multimodal behavior of items' ratings.
	Hence it forces all users to subjugate to the average, that in turn hinders the rise of a multitude of preferences/tastes/tendencies.
	\begin{figure}
		\centering
 		\includegraphics[width=0.55\textwidth]{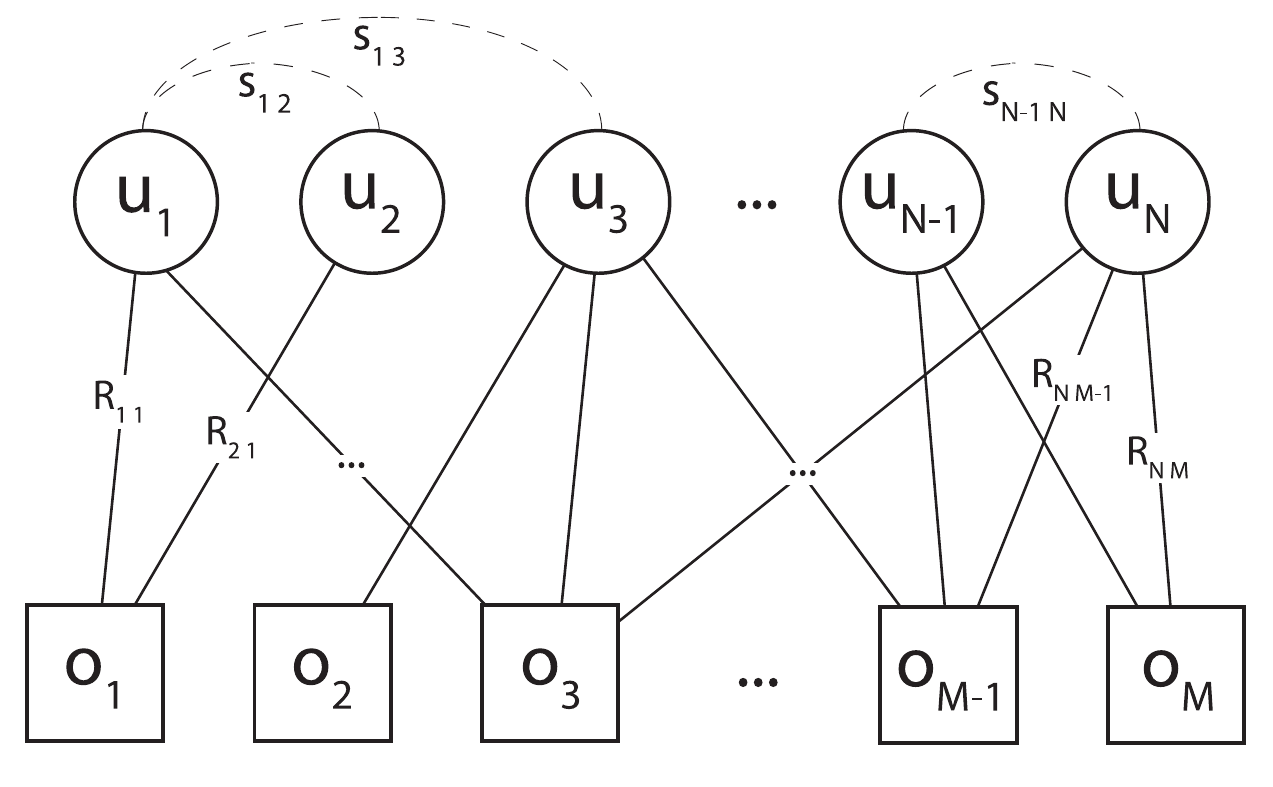}
	 	\caption{Graph representing $N$ users, $u_i$, $M$ items, $o_j$, and the ratings given by user $i$ to item $j$, $R_{ij}$. 
				 The lines represent the connection, through ratings, from the users to the items.
				 The dashed lines represent the links between users, through their similarities, $s_{mn}$.}
	 	\label{fig:grafo}
	\end{figure}
	 
	In prior work~\cite{symeonidis2011product}, the authors extended the bipartite graph approaches.
	They used implicit social networks, known as online Social Rating Networks (SRN) (that emerge from different users commenting in a similar way on a given set of products) and explicit social networks (built by users themselves, through friendship or working relation) to predict ratings and recommend products. 
	The subnetworks between users are represented in dashed lines of Figure~\ref{fig:grafo}.
	% represent implicit relations between users. 	
	
	%By exploring the former ideas, we propose an iterative reputation-based algorithm on implicit-generated rating subnetworks.
	\textbf{Our contribution.} % refer that we prove convergence and efficiency for a larger class of procedures**
	Exploring previous ideas, we propose a class of iterative reputation-based ranking system on multipartite graphs. % implicit-generated rating subnetworks.
	%Using similarity measures, our algorithm finds subnetworks within the main pool of users solely from their rating similarities, and then computes different items' rankings for each different subnetwork.
	We first prove the convergence and efficiency for a generic class of reputation-based ranking systems on bipartite graphs, that extends the results of previous works.
	After, we use similarity measures to design a system that clusters users solely based on their rating similarities, and then computes (possible) different rankings for same items on different subnetworks.
	This enables us to present custom-built rankings of items to each cluster (community). 	
	Our approach not only adapts better to the preferences of similar users, but also improves robustness against spurious users or spamming/malicious attacks.
	Further, it embeds the multimodal behavior of ratings' distribution.
	This contrasts with the algorithms that we overviewed, because those neglect the smaller sub-groups that do not identify with the majority. %(and that can be seen as outliers).
	We propose two novel similarity measures, the linear distance (LD) and the Kolmogorov distance (KD), and we test the normalized compression distance (CD) proposed in \cite{li2004similarity}. 
	Not only our approaches perform better, but also they have smaller computational complexity than CD.  
	A distinguishable feature of the LD comparing to KD is that the former responds better to noisy spam whereas the latter is slightly more robust to targeted attacks to a set of items.
	Both carry the same order of computational complexity, although KD is faster, in our implementation, than LD.
	Finally, using LD we obtain better robustness results than state-of-the-art approaches.
	%We compute this using similarity measures between users. 
	
	% The remaining of the paper is organized as follows.
	\textbf{Paper structure.} The paper is organized as follows. 
	In Section~\ref{sec:reputation_based_ranking_algorithms}, we describe and prove the convergence and efficiency of a class of reputation-based ranking iterative algorithms.
	We propose a new reputation-based ranking system in Section~\ref{sec:implementation}, where we also prove its convergence.  
	Our new algorithm is in Section~\ref{sub:bi_partite_graph_algorithms}, and its implementation is explained in Section~\ref{sec:implementation}.
	The experimental setup is described in Section~\ref{sec:experimental_setup} and we discuss our results in Section~\ref{sec:experimental_results}.
	We end the paper with final conclusions in Section~\ref{sec:conclusions}.	
% section introduction (end)
% \section{Related work} % (fold)
% \label{sec:related_work}
% 	** small survey previous works **
% % section related_work (end)
% \section{Main results} % (fold)
% \label{sec:main_results}
% 	* Identify the trend setters within network groups (more reputation)
%
% 	* The number of sub-networks must depend on some criteria and show maximize/minimize some criteria
%
% 	* A networks that is formed by the users induce another network on the products. Compute this and extract info from that.
% % section main_results (end)
\section{Reputation-based ranking algorithms} % (fold) 
\label{sec:reputation_based_ranking_algorithms}
	A reputation-based ranking system assigns a reputation to each user and then utilizes it to weigh their individual ratings on products in order to compute the products' rankings.
	Let $U$ be a set of users, $O$ a set of items, $R_\bot,R_\top$ the minimum and maximum ratings, respectively, with $\mathcal R =[R_\bot,R_\top]\cap \mathbb Z^+$ the set of strictly positive integers, the allowed ratings, $\Delta_\mathcal R=R_\top-R_\bot$, and $R\subseteq U\times O\times \mathcal R$ be the set of ratings given by users to items.
	For instances, if user $i$ rates item $j$ with rating $R_{ij}$, then we write it as $(i,j,R_{ij})\in U\times O\times\mathcal R$, or simply $R_{ij}\in\mathcal R$.
	% and denote the rating that user $i$ gave to item $j$ as $R_{ij}\in\mathcal R$, whenever user $i$ rated item $j$.
	We denote the \emph{set of items rated by user} $i$ as $O_i=\{o_j|\exists {R_{ij}}\in\mathcal R$~s.t.~$(u_i,o_j,R_{ij})\in R\}$, the \emph{set of users that rated item} $j$ as $I_j = \{u_i|\exists {R_{ij}}\in\mathcal R$~s.t.~$(u_i,o_j,R_{ij})\in R\}$.
	From now on, we consider normalized ratings, $0<R_{ij}\leq 1$, therefore the rankings and reputations take values in $]0,1]$.
	
	We model the ranking system on a weighted graph $\Gg=(U\cup O,R)$. % ou bipartite graph $\mathcal B=(U,O,R)$}
	We first consider a bipartite graph, $\Bb$, with the two sets of vertices given by users and items, see the subgraph in Figure~\ref{fig:grafo} when only considering edges between users and items.
	Subsequently, we consider an extra set of edges between the users' vertices, connecting users that are similar, see Figure~\ref{fig:grafo}.
	%
	% We first consider a bipartite graph, $\Bb$, where we encapsulate the network within users on one side, and the items on the other side, see Figure \ref{fig:grafo}.
	% Then we consider the subgraph, among the users in the graph, induced by the ratings give to common items.
	This gives rise to a multipartite graph, $\Mm$. 
	A multipartite graph is a graph such that that to color vertices with a common edge using different colors with need more than two colors, see~\cite{bollobas2013modern}. 
	Here, we need one color for the items and at least two more whenever there is a cluster with more than one user.
%	Because, in order to color vertices with a common edge, using different colors, we need one color for the items and at least two more whenever there is a cluster with more than one user, see~\cite{bollobas2013modern}.	
	This can either model users' networks generated by users themselves, like social networks see \cite{symeonidis2011product}, or, as we propose, automatically generated by the ranking systems based on ratings of items given by users.
	%We don't address the former and develop the latter in Section \ref{sub:implicit_user_based_subnetworks}.
	A partition on subnetworks may also be done on the items' side, although we do not explore this possibility.
	\subsubsection{Bipartite graph algorithms} % (fold)
	\label{sub:bi_partite_graph_algorithms}
	Here, we generalize the iterative reputation-based ranking methods, discussed above, that take the form of:
	\begin{equation} \label{eq:rbra}
		\begin{cases}
			 r^{k+1} &= g_R(c^k) \\ % \frac{1}{|I_j|} \sum_{u_i\in I_j}R_{ij}c_i^k\\
			 c^{k+1} &= h_R(r^{k+1})%, \quad c_i^0 \in ]0,1], %1-\frac{f(\lambda,O_i)}{|O_i|} \sum_{o_j\in O_i}|R_{ij}-r_{j}^{k+1}|^p.
		\end{cases},
	\end{equation}
	where $k$ denotes the iteration index and $c^{0}$ the vector of initial reputation of users %Coupled with the initial vector of users' reputation $c^0$, 
	with $c_i^0\in ]0,1]$.
	Here, $r=(r_1,\hdots,r_{|O|})$ with $r_j$ denoting the ranking of item $j$, computed with the function $g_R:[0,1]^{|U|}\to [0,1]^{|O|}$, with the set of ratings, $R$, as a parameter.
	The users' reputations, denoted as $c=(c_1,\hdots,c_{|U|})$, where $c_i$ is the reputation of user $i$, are determined by the function $h_R:[0,1]^{|O|}\to [0,1]^{|U|}$.	
	
	Next, we present more general results that subsume all the proofs of convergence and efficiency in~\cite{li2012robust}. 
	This allows to design a wider range of convergent and efficient reputation-based ranking systems.
	Consider a Banach space, $\mathcal X$, with an induced distance $d:\mathcal X\times \mathcal X \to [0,1]$. 
	Using the Lipschitz condition~\cite{kreyszig1989introductory}, we prove the following results.
	\begin{lemma}
		Consider the iterative scheme in~\eqref{eq:rbra}. 
		Let $g_R$ and $h_R$ be $\eta_g$ and $\eta_h$-Lipschitz maps, respectively. 
		It follows that $g_R \circ h_R$ is an $\eta$-Lipschitz map, with $\eta = \eta_g \eta_h$.
		If $\eta<1$, then~\eqref{eq:rbra} is a contraction.
	\end{lemma}
	\begin{proof}
		Since the domain of $g$ contains the codomain of $h$ and both are Lipschitz the composition, $g\circ h$, is also Lipschitz.
		Let $d$ be a distance, we prove the induction's basis:
		\begin{equation*}
			\begin{split}
				d(r^2,r^1) &= d\lb g_R(c^{1}),g_R(c^0) \rb \\
						   &= d\lb (g_R\circ h_R)(r^{1}),(g_R\circ h_R) (r^0)\rb \\
						   &\leq \eta d(r^{1},r^0),
			\end{split}
		\end{equation*}
		where $\eta\in[0,1[$ is the Lipschitz constant for $g_R \circ h_R$.
		The induction step then reads
		\begin{equation*}
			\begin{split}
				d(r^n,r^{n-1}) &= d\lb g_R(c^{n-1}),g_R(c^{n-2}) \rb \\
							   &= d\lb (g_R\circ h_R)(r^{n-1}),(g_R\circ h_R) (r^{n-2})\rb \\
							   &\leq \eta d\lb r^{n-1}, r^{n-2} \rb 
							   = \eta d\lb g_R(c^{n-2}),g_R(c^{n-3}) \rb \\
							   &\leq \eta^{n-1}d(r^{1},r^0),
			\end{split}
		\end{equation*}
		and the last inequality holds by the induction hypothesis.
	\end{proof}
	Because we are working in a Banach space the algorithm~\eqref{eq:rbra} converges to a unique value.
	Using the previous lemma, we prove the following result:
	\begin{theorem}
		The class of iterative reputation-based ranking algorithms~\eqref{eq:rbra} converges.
	\end{theorem}
	\begin{proof}
		Let $m,n \in \Nn$. 
		For any $\varepsilon>0$, there exists an order, $N$, from which $\eta^N < (1-\eta)\varepsilon/d(r^1,r^0)$.
		Using the triangle inequality we have
		\begin{equation*}
			\begin{split}
				d(r^n,r^m) &\leq \sum_{k=m+1}^n d(r^k, r^{k-1}) \leq \sum_{k=m+1}^n \eta^{k-1} d(r^1,r^0) \\
						   &\leq \eta^m d(r^1,r^0)\sum_{k=0}^{+\infty} \eta^k \leq \eta^N\frac{d(r^1,r^0)}{1-\eta} < \varepsilon,			
			\end{split}
		\end{equation*}
		since $0<\eta<1$, therefore the algorithm~\eqref{eq:rbra} converges.
	\end{proof}	
	\begin{theorem}
		Let $d$ be a normalized distance, $d:X\to [0,1]$.
		Then the algorithm~\eqref{eq:rbra} has an exponential rate of convergence. 
	\end{theorem}	
	\begin{proof}
		The basis of the induction reads:
		\begin{equation*}
			\begin{split}
				d(r^*,r^1) &= d\lb g_R(c^{*}),g_R(c^0) \rb	\\
						   &= d\lb (g_R\circ h_R)(r^{*}),(g_R\circ h_R) (r^{0})\rb \\
						   &\leq \eta d\lb r^{*},r^0 \rb\leq\eta.
			\end{split}
		\end{equation*}
		Assume that the induction hypothesis holds, for $k=n$, then it follows that 
		\begin{equation*}
			\begin{split}
				d(r^*,r^{n+1}) &=  d\lb (g_R\circ h_R)(r^{*}),(g_R\circ h_R) (r^{n}) \rb	\\
							   &\leq \eta d\lb r^{*}, r^{n} \rb \leq \eta^{n+1} d\lb r^{*}, r^{0} \rb %\\
							   \leq \eta^{n+1}.
			\end{split}
		\end{equation*}
	\end{proof}
	To attain an error of, at most, $\varepsilon>0$, we need $\kappa=\log_{\eta}\varepsilon$ iterations, where $\eta$ is the Lipschitz constant of $g_R\circ h_R$.
	% subsection bi_partite_graph_algorithms (end)
	\subsubsection{Similarity Measures} % (fold)
	\label{sub:similarity_measures}
		To group users, according to their preferences/tastes, we need to quantify how similar they are.
		For each pair of users, that have, at least, one rated item in common, we compute a similarity distance, based on item-rating information.
		We specify three different similarity distances: one linear and two non-linear.
		In the following, let $I_{u,v}=I_u \cap I_v$ denote the set of items that both users $u$ and $v$ rated.
		Further, for each user, $u$, we denote by $\tilde u$ the string composed by the concatenation of the pairs $(item, rating)$ of its rated items.
		
		\textbf{Linear  distance.}
			\emph{We define the \emph{linear distance} as: $\mathit{LD}(u,v)=0$ if $I_{u\cap v} = \varnothing$, and otherwise
			\begin{equation*}
				\mathit{LD}(u,v) = \ell(|I_{u,v}|)\left[1- \frac{1}{|I_{u,v}|} \sum_{i\in I_{u,v}}\frac{|R_{ui} - R_{vi}|}{\Delta_\mathcal R}\right], 
			\end{equation*}
			where $R_{jk}$ is the rating that user $j$ gave to item $k$ and the function $\ell:\mathbb Z^+\to [0,1]$ penalizes on how confident we are in the users' similarity.}
			
		We propose two compression-similarity distances based on Kolmogorov complexity, see~\cite{cover2012elements}.
		Given the description of a string, $x$, its \emph{Kolmogorov complexity}, $K(x)$, is the length of the smallest computer program that outputs $x$.
		In other words, $K(x)$ is the length of the smallest compressor for $x$. 
		Let $C$ be a compressor and $C(x)$ denote the length of the output string resulting from the compression of $x$ using $C$.
		Although the Kolmogorov complexity is non-computable, there are efficient and computable approximations by compressors. % such that $|K(x)- C(x)|\leq |x|+\mathcal O(\log|x|)$.
		%Let $u$ and $v$ denote a description of items and the respective ratings given by two different users. 
			
		\textbf{Kolmogorov distance.}
			\emph{We define the \emph{Kolmogorov distance} as: $\mathit{KD}( u, v)=0$ if $I_{u,v} = \varnothing$, and otherwise
			\begin{equation*}
				\mathit{KD}( u, v)=\left(1+|C(\tilde u)-C(\tilde v)|\right)^{-1}.
			\end{equation*}}
			If we think of $\tilde u$ as a training set and the compressor $C$ as a taste/preferences profiler of the users, then we can aggregate users $u$ and $v$ by computing the difference of their preferences $|C(\tilde u)-C(\tilde v)|$. 
			
		\textbf{Compression distance.}
		 	\emph{Based on the \emph{normalized compression distance}, see \cite{li2004similarity}, we define the \emph{compression distance} as: $\mathit{CD}(u, v)=0$ if $I_{u,v} = \varnothing$, and otherwise
		 	\begin{equation*}
		 		\mathit{CD}( u, v) = 1 - \frac{ C(\tilde u\tilde v) - \min\{C(\tilde u),C(\tilde v)\} }{ \max\{C(\tilde u),C(\tilde v)\} },	
		 	\end{equation*}
		 	for the string $\tilde u\tilde v$, the concatenation of $\tilde u$ and $\tilde v$.}
		%For users that share, at least one, rated item note that the $\mathit{LD}$ and $\mathit{CD}$ distances only evaluate the information contained in the common rated items. % that both users rated, 
		%Whilst $\mathit{KD}$ accounts for the information of all rated items that each user rated.
		%This marks a clear distinction between those distances that we explore in the experimental results.
		% paragraph compression_similarity_distance (end)
	% subsection similarity_measures (end)
	\subsubsection{Multipartite graph algorithms} % (fold)
	\label{sub:implicit_user_based_subnetworks}
		Consider a similarity distance, $\mathit{SM}$. 
		We group users in communities using $\mathit{SM}$.
		For a specified affinity level threshold, $\alpha$, we set $S_{u,v}=1$ if $\mathit{SM}(u,v)>\alpha$ and $0$ otherwise, where $S$ is the (possible sparse) adjacency matrix, that characterizes the undirected graph $\Mm \equiv \Mm(S)$.
		A bigger $\alpha$ means that users need to be more strongly related in order to be connected, which translates to a finer granularity in the communities. 
		We compute the subnetworks of $\Mm$, $\Mm_i$ for $i\in\mathcal I$, that are the connected components of $\Mm$. 
		Then, we select the most relevant subnetworks $\{\Mm_j\}_{j\in \mathcal J}$, with $\mathcal J\subseteq \mathcal I$. 
		Subsequently, we apply the multipartite-graph algorithm to each subnetwork $\Mm_j$, for $j\in\mathcal J$ to compute the reputation of users and the ranking of items.
	% subsection implicit_user_based_subnetworks (end)
	
	Let $\mathit{rep\_rank}$ denote a reputation-based ranking algorithm~\eqref{eq:rbra}.
	We summarize the previous steps in the clustering reputation-based ranking Algorithm~\ref{alg:Rankingsubnet}.
	\begin{algorithm} 
		\caption{Clustering reputation-based ranking algorithm.}
		\label{alg:Rankingsubnet}
		\begin{algorithmic}[1]
			\STATE{\textbf{input}: $\alpha$, \emph{dataset}}
			\STATE{\textbf{build} $S$ from dataset and apply threshold $\alpha$}
			\STATE{\textbf{build} $\Mm\equiv\Mm(S)$}
			%\STATE{\textbf{apply} threshold \alpha}
			\STATE{\textbf{find} the connected components of $\Mm$, $\{\Mm_i\}_{i=1}^m$}
			% \FOR{$i$ from $1$ to $m$}
			% 	\STATE{\textbf{compute} $\mathit{rep\_rank}(\Mm_i)$}
			% \ENDFOR
			\STATE{\textbf{output}: weighted average of $\{\mathit{rep\_rank}(\mathcal M_i)\}_{i=1}^m$}
		\end{algorithmic} 
	\end{algorithm}
	%\subsection{Implicit item-based subnetworks} % (fold)
	%\label{sub:implicit_item_based_subnetworks}
	% subsection implicit_item_based_subnetworks (end)
	%\subsection{General multipartite graphs} % (fold)
	%\label{sub:general_multipartite_graphs}
	%	use link between items and links between users to suggest new products with better efficience 
	% subsection general_multipartite_graphs (end)
% section reputation_based_algorithms (end)
\section{Implementation} % (fold)
\label{sec:implementation}
	From now on, we will consider algorithms defined by the equations~\eqref{eq:g} and~\eqref{eq:h} below.
	We compute the ranking of the object, $r_j$, as a weighted average.
	That is, the ratings of user $i$ to item $j$ are weighted by the users' reputation, $c_i$, and therefore $g_R$ in~\eqref{eq:rbra} becomes:
	\begin{equation}\label{eq:g}
		r_j^{k+1}= \sum_{u_i\in I_j}R_{ij}c_i^{k+1}\bigg/\sum_{u_i\in I_j}c_i^{k+1}.
	\end{equation}
	For $h_R$ in~\eqref{eq:rbra}, we tested three different functions parametrized by $f_{\lambda,s}$:
	% \begin{equation}\label{eq:h}
	% 	c_i^{k+1}=1-f_{\lambda,s}(O_i) e_{R,i}(r),
	% \end{equation}
	% where
	% \begin{equation*}
	% 	e_{R,i} = \begin{cases}
	% 		\frac{1}{|O_i|} \displaystyle\sum_{o_j\in O_i}|R_{ij}-r_{j}^{k}|^p \\
	% 		\underset{o_j\in O_i}{\max}|R_{ij}-r_{j}^{k}|^p \\
	% 		\underset{o_j\in O_i}{\min}|R_{ij}-r_{j}^{k}|^p
	% 	\end{cases}.
	% \end{equation*}
	\begin{equation}\label{eq:h}
		c_i^{k+1}=1-f_{\lambda,s}(O_i) \cdot
		\begin{cases}
			\frac{1}{|O_i|} \displaystyle\sum_{o_j\in O_i}|R_{ij}-r_{j}^{k}|^p \\
			\underset{o_j\in O_i}{\max}|R_{ij}-r_{j}^{k}|^p \\
			\underset{o_j\in O_i}{\min}|R_{ij}-r_{j}^{k}|^p
		\end{cases}.
	\end{equation}
	The users' reputation is chosen as a function of the average, maximum or minimum disagreement of individual user's ratings, $R_{ij}$, and the rankings of the rated items, $r_j$.
	In order to control the penalization a user incurs on, for not rating according to the ranking, we define a \emph{decay function} $f_s$. 
	We consider three different decay functions:\\[0.1cm]
		i) $f_{\lambda,s}^1(x) = \lambda$\qquad\qquad ii) $f_{\lambda,s}^2(x) = \lambda \left(1 - e^{-\frac{x}{2}}\right)$\\iii) $f_{\lambda,s}^3(x) = \lambda \left[ 1-(1-\upsilon)(1+e^{s-x})^{-1} \right]$,\\[0.1cm]
	where $\lambda \in[0,1[$, $\upsilon\in ]0,1[$ is the lowest penalization an user can incur
	and $s \in \Nn$ is a parameter based on the number of rated items such that the penalization is decreased by a half.	
	The role of the decay function is to control the penalization a user $i$ suffers if it doesn't rate the item, $R_{ij}$, close to its ranking $r_j$.
	The first, constant, function $f_{\lambda,s}^1$ above is proposed in \cite{li2012robust}, the second is an exponential decrease function, and the third is a logistic function.
	In the second and third cases the penalization increases and decreases, respectively, with the number of rated products.
	In the remaining of the paper we fix for $h_R$ the average and for $f_{\lambda,s}$ the constant function, $f_{\lambda,s}^1$, denoting by \emph{bipartite weighted average} (BWA) the resulting iterative scheme in equations~\eqref{eq:g} and~\eqref{eq:h}.
	\subsubsection{Convergence} % (fold)
	\label{ssub:convergence}
		Here, we prove the convergence of the proposed method.
		In what follows, for a given vector $x\in\mathbb R^n$ and $p\in\mathbb Z^+$, the $p$-norm of $x$ is $\|x\|_p=(\sum_{i=1}^n|x_i|^p)^{\frac{1}{p}}$, and the $\infty$-norm is $\|x\|_{\infty} = \max_{i\in\{1,\ldots,n\}}|x_i|$.
		\begin{lemma}\label{prop:convergence}
			For all $\lambda\in[0,(1+\Delta_R)^{-1}[$, the iterative method in~\eqref{eq:rbra} with functions $g_R$ and $h_R$ defined as in~\eqref{eq:g} and~\eqref{eq:h} converges.
		\end{lemma}
		\begin{proof}
			Between iterations, $r^{k+1}$ and $r^k$, we get
			\begin{equation*}
				\|r_j^{k+1}-r_j^k\|_{\infty} = \left\| \frac{R_{j} \cdot c^{k+1}}{\|c^{k+1}\|_1} -  \frac{R_{j} \cdot c^{k}}{\|c^{k}\|_1} \right\|_{\infty}.
			\end{equation*}
			Without loss of generality, assume that $\|c^{k+1}\|_1 \geq \|c^{k}\|_1$, then the above difference is equal to
			\begin{equation*}
				\begin{split}
				%\|r_j^{k+1}-r_j^k\|_{\infty} &= 
					\left\| \frac{R_{j} \cdot c^{k+1}}{\|c^{k+1}\|_1}  - \frac{R_{j} \cdot c^{k}}{\|c^{k+1}\|_1} + \frac{R_{j} \cdot c^{k}}{\|c^{k+1}\|_1} - \frac{R_{j} \cdot c^{k}}{\|c^{k}\|_1} \right\|_{\infty}\\
					\leq \left\| \frac{R_{j} \cdot c^{k+1}}{\|c^{k+1}\|_1} - \frac{R_{j} \cdot  c^{k}}{\|c^{k+1}\|_1} + \frac{R_{j} \cdot  c^{k}}{\|c^{k}\|_1} - \frac{R_{j} \cdot  c^{k}}{\|c^{k}\|_1} \right\|_{\infty} \\
					\leq \frac{R_\top}{\|c^{k+1}\|_1} \left| c_{\gamma}^{k+1}-c_{\gamma}^k \right|,
				% &= \frac{R^T}{\|c^{k+1}\|_1} \left\| c^{k+1}-c^k \right\|_1
				\end{split}
			\end{equation*}
			where $\left| c_{\gamma}^{k+1}-c_{\gamma}^k \right| = \max_{i\in I_j} \left| c_{i}^{k+1}-c_{i}^k \right|$.
			The iteration step for the reputation, $c$, gives us
			\begin{equation*}
				\begin{split}
					|c_i^{k+1}-c_i^k| %=\left| 1- \frac{f(\lambda,O_i)}{|O_i|} \sum_{i} \left| R_{ji} - r_j^{k+1} \right|^p - 1 + \frac{f(\lambda,O_i)}{|O_i|} \sum_{i} \left| R_{ji} - r_j^{k} \right|^p\right| \\
						&\leq \frac{|f_{\lambda,s}(O_i)|}{|O_i|} \sum_{j} \left|  \left| R_{ji} - r_j^{k} \right|^p - \left| R_{ji} - r_j^{k-1} \right|^p\right| \\
						&\leq \lambda |r_{\beta}^{k}-r_{\beta}^{k-1}|,
				\end{split}
			\end{equation*}
			where $\left| r_{\beta}^{k+1}-r_{\beta}^k \right| = \max_{j\in O_i} \left| r_{j}^{k+1}-r_{j}^k \right|$, and using the triangular inequality, the translation invariance of norms and the fact that $|f_{\lambda,s}(O_i)|\leq 1$.
			Combining the previous inequalities we get
			\begin{equation}\label{eq:pr:conv}
				|r_j^{k+1}-r_j^k| \leq \frac{\lambda}{\|c^{k+1}\|_1} |r_{\beta}^{k}-r_\beta^{k-1}|.
			\end{equation}
			Setting $\lambda<(1+\Delta_R)^{-1}$, since $1-\Delta_\mathcal R\lambda \leq \|c\|_1 \leq 1$, we ensure that~\eqref{eq:pr:conv} is a contraction, therefore~\eqref{eq:rbra} converges.
		\end{proof}
		In our simulations, we consider $\Delta_\mathcal R = 1-0.2$. 
		Therefore if $\lambda\leq \frac{5}{9}$ then our proof ensures that the algorithm converges.
		Notwithstanding, we ensure convergence for any $\lambda\in[0,1[$ modifying the denominator of~\eqref{eq:h} to be $\max\{\|c^{k+1}\|_1,1\}$.
	% subsubsection convergence (end)
	\subsubsection{Computational complexity analysis} % (fold)
	\label{ssub:complexity_analysis}
		The time complexity of Algorithm~\ref{alg:Rankingsubnet} is given by the sum of the complexities of each step.
		Let $\Gg=(V,E)$ denote a graph, where $V$ is a set of vertices and $E$ a set of edges.
		Step 3 consists in building $\Gg$, this is done computing its sparse adjacency matrix, $\mathcal M$, where each rating is used once. 
		Henceforth, the time complexity is $\mathcal O(|C||R|)$, where $|C|=O(1)$ for similarities LD and KD, and where, for the CD, $|C|$ is the worst case complexity of compressing the concatenation of pairs of users. 
		Step 4 can be performed using Tarjan's Algorithm~\cite{hopcroft1973algorithm}, with time complexity in the worst case of $\mathcal O(|V|+|E|)$.
		Step 5 has, in the worst case, the same time complexity of~\cite{li2012robust}, i.e., $\mathcal O(\kappa|R|)$.
		In summary, Algorithm~\ref{alg:Rankingsubnet} has worst case time complexity of $\mathcal O\left((\kappa+|C|)|R|+|V|+|E|\right)$.
		In theory $|E|$ can be, in the worst case $|U|^2$, leading to a time complexity of $\mathcal O\left((\kappa+|C|)|R|+|U|^2\right)$.
		In practice, since (often) the users are sparsely connected in $\Gg$, $|E|=\mathcal O(|U|)$, resulting in a time complexity of $\mathcal O\left((\kappa+|C|)|R|+|U|\right)$.
		In all cases the space complexity of Algorithm~\ref{alg:Rankingsubnet} is $\mathcal O(|R|)$.
	% subsubsection complexity_analysis (end)	
% section implementation (end)
\section{Experimental setup} % (fold)
\label{sec:experimental_setup}
	We run all experiments on macOS 10.12.3 with 2.8 GHz Intel Core i5 and 8 GB 1600 MHz DDR3, in MATLAB~2016a.
	\subsubsection{Datasets} % (fold)
	\label{sub:datasets}
		In this work, we use real world datasets that were obtained from the Stanford Large Network Dataset Collection, \cite{snapnets}.
		% We used the ``Amazon rating movies an TV'' dataset, with $XXX$ users, $YYY$ items and $4,607,047$ ratings.
		% We also used the  ``Automotive'' dataset, with $XXX$ users, $YYY$ items and $1,373,768$ ratings.
		We used the $5$-core version of ``Amazon Instant Video'' dataset that consist of users that rated at least $5$ items, as in~\cite{mcauley2015inferring}.
		It has $5,130$ users, $1,685$ items and $37,126$ ratings, with $R_\bot=1$ and $R_\top=5$.
		These datasets include only tuples of the form (user, item, rating, timestamp) and no metadata or reviews.
		We normalized all ratings, dividing them by $R_{\top}$.	
		% $426,922$ users, $23,965$ items and $583,933$ ratings.
		% The majority of the users ($\approx 70\%$) in this dataset only rated a single item.
		% This may not allow us to fully explore the advantages of the decay functions and the clustering.
		% Thereupon, we used the $5$-core version that consist of users that rated at least $5$ items.
		% This dataset has $xxx$ users, $yyy$ items and $37,126$ ratings.
	% subsection datasets (end)	
	\subsubsection{Benchmarks} % (fold)
	\label{sub:benchmarks}
		We compare our results with the reputation-based ranking system in~\cite{li2012robust}.
		The authors already compared their algorithm with the state-of-the-art algorithms. % several other algorithms.
		Namely, the HITS~\cite{kleinberg1999authoritative}, the Mizz~\cite{mizzaro2003quality}, the YZLM~\cite{yu2006decoding} and the dKVD~\cite{de2010iterative} algorithms.
		Furthermore, the authors of~\cite{li2012robust} showed that their algorithm outperforms all the others, in all standard metrics.
	% subsection benchmarks (end)
	\subsubsection{Evaluation metrics} % (fold)
	\label{sub:evaluation_metrics}
		To quantitatively assess the quality of the ranking systems, we compute the Kendall rank correlation coefficient, a.k.a. \emph{Kendall's tau}, $\tau$, see~\cite{kendall1938new}.
		This statistic measures the ordinal association between two quantities.
		Intuitively, the Kendall correlation between two variables is higher when observations are identical and lower otherwise. 
		Given two sets $X$ and $Y$, let $C$ and $D$ denote the sets of concordant and discordant pairs of elements in $X\times Y$, respectively. 
		The Kendall's tau is defined as $\tau=(|C|-|D|)/(|C|+|D|)$.
		The \emph{effectiveness} is given by the Kendall tau of the rankings' vector, $r$, versus a ground truth, $\hat r$, that is $\tau(r,\hat r)$.
		The selection of the ground truth is an intricate matter, for practical reasons we choose the AA.
		Because of its simplicity, and its popularity among ranking systems, \cite{jurczyk2007discovering,de2010iterative}.
		% Because it is commonly used, for its simplicity, and it is a popular method for ranking systems, \cite{jurczyk2007discovering,de2010iterative}.
		Nonetheless, evaluating the discrepancy between the ranking vector, $r$, and the AA might not be very informative. 
		Since it does not capture the possible multimodal behavior of ratings, therefore might not be very useful to evaluate the quality of a ranking system.
		When using $f_s(x)=\lambda=0.1$, as in~\cite{li2012robust}, we see that it yield high effectiveness values ($\tau \approx 1$).
		The reason is that in this case $c\in[0.9,1]$, and we obtain rankings very close to AA, and in the case $c=1$ it yields the AA.
		This in turn plays against what we want to achieve, a departure from the simplicity of AA, in order to gain insight through an intelligent weighting of ratings.
		Therefore the effectiveness might not be that important to evaluate the quality of a ranking system.
		In such a heterogeneous environment the AA do not carry useful information. 
		So in the bipartite case we opt for the robustness metric.
		In the multipartite case, the effectiveness may be helpful, since we might want to check for homogeneity within the clusters. % classes/clusters/groups/subnetworks.
		Henceforth, we generalize the Kendall tau as
		\begin{equation*}
			\begin{split}
				\bar \tau = \frac{1}{|\Mm|} \sum_{i=1}^N |\Mm_i| \tau_{\Mm_i}, \quad
				\Mm = \bigcup_{i=1}^N \Mm_i, \quad \Mm_i \bigcap \Mm_j = \varnothing, 
			\end{split}
		\end{equation*}
		where $\Mm_i$ is a subgraph of $\Mm$, with $|\Mm_i|$ vertices. % denotes the number of vertices of the graph.
		The effectiveness of a cluster, $\Mm_i$, is denoted by $\tau(r_{\Mm_i},r_{AA|\Mm_i})$. %, measuring its homogeneity.
		
		The \emph{robustness} evaluates the ability of the algorithm to cope against noise or spamming attacks.
		A noisy user gives random ratings to a random set of products, see~\cite{aggarwal2016recommender}.
		Whilst a spamming attacker targets a set of items with the intent of increasing (\emph{Push Attack}) or decreasing (\emph{Nuke Attack}) their rankings.
		In our simulations, we run our algorithm in the original dataset and in spammed ones.
		After, we compute the Kendall tau between the rankings of the original dataset, that we assume as the ground truth, against the rankings obtained with the spammed dataset.
			
		For the multipartite approach we compute the Kendall tau as $\bar \tau = \tau(\bar r,\bar r_{\text{spam}})$, where $\bar r$ is the vector of $\bar r_j$'s given by
		\begin{equation*}
			\bar r_j = \frac{1}{|\hat \Mm|} {\sum_{m} |\hat \Mm_m| r_j^m}, \,\,\text{where }\,\, \hat \Mm = \bigcup_m \hat \Mm_m
		\end{equation*}
		is the union of subnetworks where users rated item $j$.
		%follows:
		% \begin{equation*}
			% \begin{split}
			% 	\Mm = \bigcup_{i} \Mm_i, \quad \Mm_i \bigcap \Mm_j = \varnothing, \quad
			% 	\bar \tau = \frac{1}{|\Mm|} \sum_{i=1}^N |\Mm_i| \tau_{\Mm_i},
			% \end{split}
		% \end{equation*}
		% where $\Mm_i$ are the subgraphs of $\Mm$, and $|\cdot|$ is the number of vertices of the graph.
		The effectiveness within a subgroup, $\Mm_i$, is given by $\tau(r_{\Mm_i},r_{AA})$, and it measures the homogeneity of the it.
		This measure is useful to assess the quality of the partition of the original network, and it can be used to tune the affinity level, $\alpha$, between users so that they are in the same cluster. % in the same community. 
	% subsection evaluation_metrics (end)
% section experimental_setup (end)
\section{Experimental results} % (fold)
\label{sec:experimental_results}
	Here we discuss the obtained results and evaluate our ranking algorithm using the robustness metric.
	\subsubsection*{Robustness} % (fold)
	\label{ssub:robustness}
		\begin{figure}[ht]     
			\centering
		        \includegraphics[width=\textwidth]{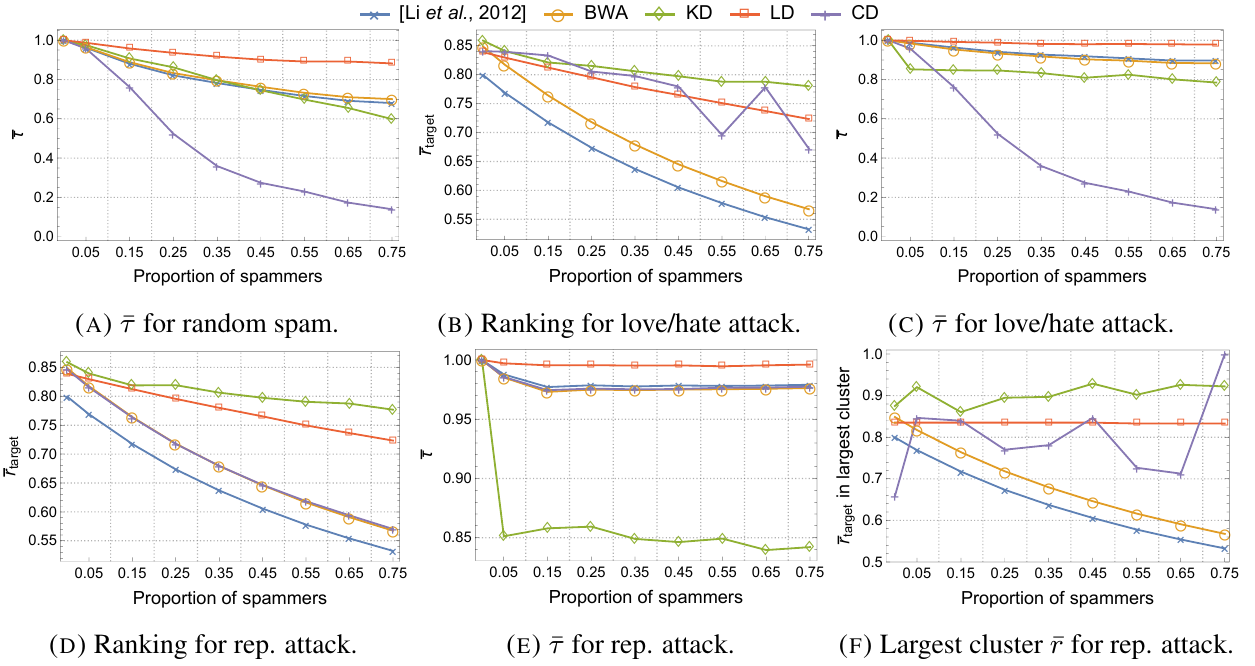}
				\caption{The plots depict the variation of the Kendall $\bar \tau$ and of the ranking of the targeted item, $\bar r_{\text{target}}$ (the most voted), with respect to the fraction of spammers from $0$ to $0.75$.
				For the random spam, (a), we consider the fraction of the total number of users.
				For the love/hate attack, (b)-(c), and for the reputation attack, (d)-(e), we consider the fraction of the raters of $\bar r_{\text{target}}$.
				Figure (f) shows the variation of $\bar \tau_{\text{target}}$ in the biggest cluster along no clustering systems.}
				\label{fig:results}
		\end{figure}
		
		The robustness is a performance measure that evaluates the ability of a ranking system to withstand spamming or ranking attacks.
		In order to mitigate the effect of spammers/malicious users, we first propose using the function $f_s$ in the iterative method~\eqref{eq:rbra}.
		A second way to mitigate the spamming effect is grouping users into communities, as explained in Section~\ref{sub:implicit_user_based_subnetworks}. 
		
		In the bipartite graph scenario, the information available to a new user is every products' rankings.
		This information can be used by malicious users to tamper with the ranking of an item in a malicious way (push or nuke it.)
		For instances, in a reputation-based systems, an attacker can give ratings matching the ranking of items to increase its reputation, before attacking an item.
		
		Whereas when allowing for communities (subnetworks), either the user is already classified into a community and he access the item's ranking within that community, or he is a new user. 
		In this case, the displayed ranking, $\bar r_j$, of the item, $j$, is the weighted average of its ranking within each subnetwork.

		Both of these scenarios mitigate the spamming effect. 
		Since the information made available is not a sufficient statistic, a user cannot fully recover all the information to attack, in a more efficient way, the underlying ranking system.
		
		In this section we discuss the robustness of the algorithm to different kinds of spamming/attacks:
		\begin{itemize}
			\item \emph{Random spamming}: A set of spammers give random ratings, uniformly distributed on $\mathcal R$, to a random number of items, following a Poisson distribution, starting at $1$ with parameter $\lambda_{P} = 5$. % $\lambda_{P} = 1$, for the first dataset **ref** and  for the second dataset **ref**.
				The rated items are randomly selected following the distribution of the number of ratings per item, in the initial dataset.
			\item \emph{Love/hate attack}: A set of spammers targets one item to push/nuke and selects another set of items to nuke/push.
				In our simulations, each attacker nukes the most voted item and pushes another random set of nine filler items. 
			\item \emph{Reputation attack}: In this case, a set of spammers targets one item to push/nuke its ranking. 
				They randomly select another fixed number of items, from the initial dataset, typically the most popular ones, and give them the closest ratings to their rankings. 
				We discuss the nuke version, because the push attack is similar.
			% \item Add spamming users that give random ratings to a random set of items.
			% \item Use a bot to nuke one item's rating 
		\end{itemize}
		In all experiments we set $\lambda=0.3$, $\alpha=0.8$, and for the $\text{LD}$ method the \emph{confidence level function} $\ell(|I_{u,v}|)=\theta^{-1}$ if $|I_{u,v}|\leq \theta$ and $1$ otherwise.
		The parameter $\theta$ sets the number of common rated items of users $u$ and $v$ from which we are confident that they can be similar. We choose $\theta=3$. 
		To evaluate the effect of the random spamming, we compute the Kendall tau, $\tau(\bar r,\bar r_{spam})$.
		In the bipartite case, we see that BWA copes slightly better with noise than the algorithm proposed in~\cite{li2012robust}, which is claimed to be the best procedure with convergence guarantees.
		Using the multipartite graph systems, we notice an increase of robustness for the LD, and for the KD a similar robustness to the bipartite methods.
		The CD performs the worst, because it is sensitive to accommodating new users by rearranging the clusters, which in turn degrades the $\bar \tau$, see Figure~2a.
		
		% To evaluate the performance of the algorithms we use the Kendall tau metric.
		% In the bipartite graph case, we use the non-spammed rankings, $r$, as ground truth $\tau(r,r_{\text{spam}})$, in the multipartite method, we compute $\bar \tau = \tau(\bar r,\bar r_{\text{spam}})$. %\frac{1}{|\Mm|} \sum_i |\Mm_i| \tau(r_i,r_i^{\text{spam}})$.
		% We then compare the robustness of the Robust Reputation-Based Ranking on Bipartite Rating Networks algorithm, against algorithms using multipartite graphs.
		% \begin{figure*}
	% 	\centering
	% 	\begin{subfigure}[b]{0.29\textwidth}
	%  	        \includegraphics[width=\textwidth]{Figures/tau_rob_1cut.pdf}
	%  	        \caption*{Random attack}
	%  	        %\label{fig:obj1}
	%  	\end{subfigure}
	%  	   % \quad % add desired spacing between images, e. g. ~, \quad, \qquad, \hfill etc.
	%  	\begin{subfigure}[b]{0.29\textwidth}
	%  	        \includegraphics[width=\textwidth]{Figures/tau_rob_2cut.pdf}
	%  	        \caption*{Love/hate attack}
	%  	        %\label{fig:obj2}
	%  	\end{subfigure}
	%  	% \quad % add desired spacing between images, e. g. ~, \quad, \qquad, \hfill etc.
	%  	\begin{subfigure}[b]{0.38\textwidth}
	%  	        \includegraphics[width=\textwidth]{Figures/tau_rob_3.pdf}
	%  	        \caption*{Informed attack}
	%  	        %\label{fig:obj3}
	%  	\end{subfigure}
	%  	\caption{Effect of the $\lambda$ parameter on the effectiveness metric.
	% 			 From left to right we used as decay functions $f1$, $f2$, $f3$.}
	% 	\end{figure*}

		To test the robustness (to attacks) of the methods we simulate attacks to the most voted item, $r_{\text{target}}$.
		We ranged the proportion of spammers from $0$ to $0.75$ of the fraction of total voters on the target item.
		For the love/hate attack, Figures~2b and~2c, we see that the variation of the rating $r_{\text{target}}$ was smaller, henceforth the attack is less effective, when using methods KD and LD. 
		The smallest variation of $\tau$ occurs with method LD, guaranteeing not only a strong robustness to the attack on $r_{\text{target}}$, but also to possible side effects.
		On the other hand, KD was very effective to deter the rank attack, but failed to prevent noticeable changes on other items' rankings.
		This is a consequence of the reorganization of the subnetworks to minimize the effect of the attack on $r_{\text{target}}$, and our generalization of the Kendall tau does not account for this repercussion.
		Moreover, in the larger clusters containing users who rated the targeted item, the ranking of the item was kept unchanged for both LD and CD, the same effect as for the reputation attack in Figure~2f. This indicates that the attackers are not grouped with normal users and thus do not affect the rankings of items in the cluster. For KD, although the ranking oscillates due to the reorganization of clusters, it is not nuked as in~\cite{li2012robust} and BWA. 
%		 ***
%		 Comment the variation of the $r$ within communities  
%		 ***
		
%		*****
		
		Utilizing communities, the effect of the reputation attack on the ranking of the targeted item was attenuated, see Figures~2d and~2e.
		Because the intelligent attacker chooses the closest rating to the ranking of the filler items it should not affect drastically the rankings of the filler items, and the attackers increase their reputation. 
		The ranking of the nuked item, using multipartite rating networks, drops less than when not using them. 
		%** Insert $r_{com}(posmax)$ vs $r(posmox)$ Graph **
		The organization of communities changes with the increasing number of spammers, this is an effect of the system to cope with the attack.
		Thus, it produces a bigger change in $\bar\tau$, because it reduces drastically the effect of the targeted attack, and since the ranking of the filler objects do not change drastically (the attackers rate those items with their weighted average ranking) it is not a very important side effect. 
		Again, in the larger clusters containing users who rated the targeted item,
		the ranking of the item was kept unchanged when using LD, had a small variation for KD, and had a big variation for CD. Both last two variations reflect the opposite effect on the ranking of the targeted item as what is intended by the attacker, see Figure~2f. 
		The clustering produced by KD and CD aggregate attackers with legit users (that gave smaller ratings to the target item) on a separated cluster, leaving raters who gave high on the biggest cluster.    
		Further, observe that for new users, the displayed rankings are a weighted average, whereas in each cluster they are the average within the cluster.       
%		*****
		%** To evaluate the robustness we expect the spammers/attackers to be pushed into a community of their own, or not to impact too much the rankings along groups. see the experiments first than write **
		% \begin{figure}
		% 	\centering
		% 	 		\includegraphics[width=0.48\textwidth]{Figures/rank_atk_r.pdf}
		% 	 	 	\caption{Effect of the Reputation attack on the ranking of the most voted item, $r_{622}$, with respect to the proportion of spammers ($0-0.55$) on the initial number of item's raters.}
		% 	 		\label{graph:rank_atk}
		% \end{figure}
	% subsubsection robustness (end)
% section experimental_results (end)
\section{Conclusions} % (fold)
\label{sec:conclusions}
	%Compare the robustness of our algorithm to the bipartite one.
	We develop a new multipartite ranking system that allows the coexistence of multiple  preferences by enabling different rankings for the same item for different users.
	This is achieved by automatically clustering similar users, based on their given ratings.
	For each cluster we use a bipartite reputation-based ranking system, for which we prove convergence and efficiency in a more general setting than previous results. 
	Although it favors the creation of bubbles, \emph{i.e.}, segregates users into groups, we show that our method makes the ranking system more robust to attacks and spamming.
%	To group users into communities, we used a novel approach based on the Kolmogorov complexity, the KD, that revealed very robust to attacks.
	
	As future work, we will investigate the effect of bribing users in order to influence the ranking of items, as in~\cite{grandi2016network}. 
	 And optimize and compute the cost of attacking the ranking of items for systems where an user can only rate a bought item.   
	Also, in order to reduce the rate of change in the communities, we will explore the use of steadiness functions, based on a timestamp, so that established clusters do not change so easily.
	
	We can explore a good choice of decay function in order to penalize more the reputation of users that are not similar to majority of users.	
	
	Another possible extension of the proposed algorithm is to use it for recommendation systems.
\bibliographystyle{apalike}
\bibliography{library.bib}

\end{document}